\documentclass{ifacconf}

\usepackage{graphicx}      
\usepackage{natbib}        

\usepackage{siunitx}
\usepackage{enumitem} 
\usepackage{amsmath}
\usepackage{amsfonts}
\usepackage{mathrsfs}

\DeclareMathOperator*{\argmin}{argmin}
\usepackage{booktabs}
\usepackage{multirow}
\usepackage{adjustbox} 

\newtheorem{theorem}{Theorem}
\newtheorem{corollary}{Corollary}
\newtheorem{lemma}{Lemma}

\newtheorem{definition}{Definition}

\newtheorem{assumption}{Assumption}

\DeclareMathOperator*{\plim}{plim}

\DeclareMathOperator*{\col}{}

\DeclareFontFamily{U}{mathx}{\hyphenchar\font45}
\DeclareFontShape{U}{mathx}{m}{n}{
	<5> <6> <7> <8> <9> <10>
	<10.95> <12> <14.4> <17.28> <20.74> <24.88>
	mathx10
}{}
\DeclareSymbolFont{mathx}{U}{mathx}{m}{n}
\DeclareMathSymbol{\bigtimes}{1}{mathx}{"91}
    
\newcommand{\mbb}[1]{\mathbb{#1 }}
 
\newcommand{\mcl}[1]{\mathcal{#1}}

\NewDocumentCommand{\tra}{om}{%
	\IfNoValueTF{#1}
	{#2}
	{#2_{[#1]}}%
}

\usepackage{tikz}
\usetikzlibrary{shadows,arrows.meta,positioning,backgrounds,fit,chains,scopes,calc,shapes}
\usetikzlibrary{backgrounds}
\definecolor{gray1}{rgb}{0.9,	0.91,	0.93}
\definecolor{lightblue}{rgb}{0.6, 0.76, 0.92}
\definecolor{gray2}{rgb}{0.95, 0.95,0.95	}
\definecolor{gray3}{rgb}{0.85, 0.85, 0.85	}
\definecolor{lightblue1}{rgb}{0.87, 0.92, 0.97	}
\definecolor{lightblue2}{rgb}{0.62, 0.76, 0.90	}
\definecolor{lightblue3}{rgb}{0.36	0.61	0.84}	

\newcommand{\pce}[1]{\mathsf{#1}}
\newcommand{\pcecoe}[2]{\mathsf{#1}^{#2}}
\newcommand{\relx}{(\omega)}

\newcommand{\inst}[1]{\left(#1\right)}

\newcommand{\splx}[1]{L^2\left(\Omega, \mathcal{F}, \mu; \mathbb{R}^{#1}\right)}
\newcommand{\spl}{L^2(\Omega, \mathcal{F}, \mu; \mathbb{R})}

\newcommand{\mean}{\mbb{E}}

\newcommand{\Hankel}{\mcl{H}}
\newcommand{\hankel}{\mcl{H}}

\newcommand{\dimy}{{n_y}}
\newcommand{\dimx}{{n_x}}

\newcommand{\dimz}{{n_z}}
\newcommand{\dimv}{{n_v}}

\newcommand{\I}{\mathbb{I}}
\newcommand{\N}{\mathbb{N}}
\newcommand{\T}{\mathbb{Z}}
\newcommand{\R}{\mathbb{R}}

\newcommand{\sinput}{\tilde{u}}
\newcommand{\Sinput}{\tilde{U}}
\newcommand{\UnsD}{W^\text{u}}
\newcommand{\uE}{ {V}}
\newcommand{\ue}{ {v}}

\newcommand{\lag}{\ell}

\newcommand{\Lnorm}[1]{\|{#1}\|_{L^2}}
\newcommand{\Rnorm}[1]{\|{#1}\|_{2}}

\newcommand{\zdd}{ \hankel_1\left(z^\text{d}_{[0,T-1]}\right)}
\newcommand{\ydd}{ \hankel_1\left(y^\text{d}_{[0,T-1]}\right)}

\newcommand{\udd}{ \hankel_1\left(\sinput^\text{d}_{[0,T-1]}\right)}

\newcommand{\edd}{ \hankel_1\left(\ue^\text{d}_{[0,T-1]}\right)}

\newcommand{\Bu}{ B}
\newcommand{\Bw}{ E}
\newcommand{\Du}{D}
\newcommand{\Dw}{ F}

\begin{document}
\begin{frontmatter}

\title{Uncertainty Propagation under Residual Disturbances: A Smart-Home Case Study\thanksref{footnoteinfo}} 

\thanks[footnoteinfo]{This research was partially supported by the Deutsche Forschungsgemeinschaft (DFG, German Research Foundation) – SFB 1615 – 503850735 and by the Research Council of Norway (NFR) via SARLEM (project 300172)}

\author[First]{Guanru Pan} 
\author[Second]{Dirk Reinhardt} 
\author[Second]{Sebastien Gros}
\author[First]{Timm Faulwasser}

\address[First]{
	Institute of Control Systems, Hamburg University of Technology, Hamburg, Germany (e-mails: guanru.pan@tuhh.de, timm.faulwasser@ieee.org).}
\address[Second]{Department of Engineering Cybernetics, Norwegian University of Science and Technology, Trondheim, Norway (e-mails: dirk.p.reinhardt@ntnu.no,sebastien.gros@ntnu.no)}

\begin{abstract}  
This paper presents a data-driven framework for uncertainty propagation under \textit{unmeasured or statistically unmodeled (unstructured)} disturbances. We consider \textit{residual disturbances}, which consolidate all unstructured disturbances into a single quantity that can be estimated from data. Under mild assumptions, the resulting stochastic predictor is causal and distributionally consistent, enabling efficient uncertainty quantification through polynomial chaos expansions and higher-order Chebyshev inequalities. The proposed method is validated using experimental data from a smart home in Norway. 
\end{abstract}


\end{frontmatter}

\section{Introduction}
In real-world applications, stochastic disturbances can significantly affect both the performance and safety of control systems.  Building climate control systems must cope with unpredictable weather and highly variable occupancy~\citep{Drgona2020}, while multi-energy systems are subject to fluctuations in renewable generation and energy demand~\citep{Oezmeteler2024}.  Despite advances in sensing and machine learning for disturbance forecasting, such predictions remain inherently uncertain and reliable only in a probabilistic sense.  Hence,
beyond these measurable and statistically modeled (structured) disturbances, systems are also exposed to disturbances that are \emph{unmeasured or statistically unmodeled (unstructured)}. 

Model-based uncertainty propagation relies on \emph{subspace identification} methods~\citep{VanOverschee2012}, which construct stochastic state-space realizations considering process and measurement noise. Though theoretically well-established, these methods require assumptions on the noise structure and the existence of a low-order model.
In contrast, data-driven approaches aim to treat stochastic effects directly within the behavioral framework, avoiding explicit state-space identification.  
Regularized data-driven predictors~\citep{Coulson2019,Huang2023,Breschi2023} and innovation-based projections~\citep{Wang2022} mitigate noise in Hankel matrices and initial conditions, while maximum-likelihood-based predictors~\citep{Yin2024} provide probabilistic forecasts of future trajectories.  

Building on the fundamental lemma by~\citet{Willems2005} and its stochastic extensions in~\citet{Pan2022a,Pan2023d}, which focused on systems with additive stochastic disturbances, this paper extends the framework to handle unstructured disturbances. 
In contrast to prior work, we introduce the concept of \emph{residual disturbances} to aggregate all unstructured effects into a single estimable process derived from data. 
This enables the construction of a \emph{causal} stochastic predictor with distributional consistency guarantees. 
Furthermore, we provide a tractable uncertainty propagation framework by combining Polynomial Chaos Expansion (PCE) with higher-order Chebyshev inequalities, which was not addressed in the previous work. 
The proposed approach is validated on real-world smart-home data from Norway~\citep{Reinhardt2024}. 
Parts of this work build on preliminary results reported in~\cite{Pan2025D}.

The  paper is organized as follows.
Section~\ref{sec:LTI} formulates the problem  and introduces the concept of residual disturbance. 
Section~\ref{sec:data} develops a causal and consistent data-driven predictor that incorporates residual disturbances.
Section~\ref{sec:UQ} presents tractable uncertainty propagation methods.
Section~\ref{sec:case} demonstrates the proposed framework on smart home data.
Finally, Section~\ref{sec:conclusions} concludes the paper and outlines directions for future research.

\textbf{Notation.} 
Let $(\Omega, \mathcal{F}, \mu)$ be a probability space with sample space $\Omega$, $\sigma$-algebra $\mathcal{F}$, and probability measure $\mu$. 
Let $\splx{n_v}$ denote the space of $\mathbb{R}^{n_v}$-valued random variables with finite mean and covariance. For $V \in \splx{n_v}$, denote its mean, covariance, and realization by $\mathbb{E}[V]$, $\Sigma[V]$, and $V(\omega)$, respectively. For $M \in \spl$, its $n$-th  moment  reads $
\mu_n(M) \doteq \mathbb{E}\!\left[(M - \mathbb{E}[M])^n\right]$ for $ n \ge 2.$ We define $|M| \in \spl$ by applying the absolute value pointwise, i.e., $|M|(\omega) = |M(\omega)|$.
We distinguish the deterministic $L^2$-norm of $V$, $\Lnorm{V}\doteq \sqrt{\mean[V^\top V]} \in \R^+$, from the random variable $\Rnorm{V}\in \spl : \omega \in \Omega \mapsto \sqrt{V^\top\relx V\relx} $ induced by 
 applying Euclidean norm realization-wise.
The distribution  of $V$ is defined as $\mu_V \doteq \mu(V^{-1}(A))$ for all event $A$  in the Borel set $\mcl B(\R^{\dimv})$.
For $V,\tilde V \in \splx{\dimv}$, the 2-Wasserstein distance between their distributions is 
\begin{align*}
	d_W(\mu_{V}, \mu_{\tilde V})	\doteq &\inf_{V_1, V_2 \in \splx{\dimv}}\Lnorm{ V_1 - V_2}\\
	 \text{ subject to }& V_1 \sim \mu_V,\,V_2 \sim \mu_{\tilde V}.
\end{align*}
For a finite index set $\I_{[0,T-1]}$, we write a sequence of random variables  as $V_{[0,T-1]} \doteq [V(0)^\top,\, V(1)^\top,\, \cdots,\, V(T-1)^\top]^\top,
$ with realization $v_{[0,T-1]}$.

\begin{definition}[Persistent excitation]
	Signal $v_{[0,T-1]}$ is said to be \emph{persistently exciting of order $N$} if the following Hankel matrix has full row rank:
	\[
	H_N(v_{[0,T-1]}) \doteq 
	\begin{bmatrix}
		v(0)  & \cdots & v(T-N)\\
		\vdots  & \ddots & \vdots\\
		v(N-1)  & \cdots & v(T-1)
	\end{bmatrix}.
	\]
\end{definition}

\begin{definition}[Convergence in probability] Let $\{X_T\}_{T\in \N^+} $ $\subset \splx{\dimx}$ and  $X \in \splx{\dimx}$.   $X_T$ \emph{converges in probability} to $X$, denoted by $ \displaystyle
	\plim_{T\to\infty} X_T = X,
	$
	if for all $\varepsilon > 0$,
	$\displaystyle
	\lim_{T\to\infty} \mathbb P\big(\Rnorm{X_T - X} > \varepsilon\big) = 0.
	$  
\end{definition}

\section{Problem statement and Preliminaries}\label{sec:LTI}
Consider a minimal state-space representation 
\begin{equation}\label{eq:ss_unknown_disturbances}
	\begin{aligned}
		X\inst{k+1} &= A  X\inst{k} +\Bu \Sinput \inst{k}+ \Bw \UnsD\inst{k}\\
		Y\inst{k} &= C X\inst{k} +\Du \Sinput\inst{k}+ \Dw \UnsD\inst{k},
	\end{aligned}
\end{equation}
where $ \Sinput \in (\splx{n_{\sinput}})^\mathbb Z $ collects the \textit{control inputs}  $ U$ and the \textit{structured disturbances} $W^{\text s}$, i.e., $\Sinput\doteq [U^{\top},W^{\text s\top}]^\top$. The process $\UnsD$ captures all  \textit{unstructured disturbances}, including  measurement noise.  A disturbance is termed \emph{structured} if its realizations are measured and its future statistics are known; otherwise, it is \emph{unstructured}.

For $T \in \N^+ \cup \{\infty\}$ and $\ell$ not smaller than the  lag of \eqref{eq:ss_unknown_disturbances}, consider the random-variable data set
	$
\mcl D_T \doteq \{\Sinput^{\text d}_{[-\ell,T-1]}, Y^{\text d}_{[-\ell,T-1]}\}$ with $D_T \doteq \mathcal{D}_T(\omega)$ as its  realization.  Based on $D_T$, we aim for constructing $\hat Y_{[0,N-1]} \in \splx{N\dimy}$ such that its distribution $\mu_{\hat Y_{[0,N-1]}}$ converges to the true one $\mu_{Y_{[0,N-1]}}$ if $T \to\infty$. The challenge is twofold: (i) this prediction has to be data-driven (Section~\ref{sec:data}); (ii) and it must allow a tractable  uncertainty propagation (Section~\ref{sec:UQ}).

Let  $Z(k) \doteq \begin{bmatrix}
	\Sinput_{[k-\lag,k-1]} \\
	Y_{[k-\lag,k-1]}
\end{bmatrix}$ and $\dimz = \lag(n_{\sinput}+\dimy)$,  consider the  AutoRegressive model with eXogenous input (ARX) form
\begin{equation}\label{eq:ARX_stochastic}
	Y(k) = \Xi Z(k) +D \Sinput(k) + \uE(k) ,\quad Z(0) =Z_0
\end{equation}
where the \emph{residual disturbance} $\uE$ aggregates all unstructured effects in $\UnsD$.

The ARX model~\eqref{eq:ARX_stochastic} can be written as
\begin{align*}
\hat X(k+1) &= A \hat X(k) + \Bu \Sinput(k) + H \uE(k), \quad \hat X(0)=\hat X_0,\\
Y(k) &= C \hat X(k) + \Du \Sinput(k) + \uE(k),
\end{align*}
with $H$ a deadbeat observer gain  such that $(A-HC)^\ell =0$. This reformulation resembles the \emph{innovation form} in stochastic subspace identification~\citep{VanOverschee2012}. 
The key difference is the observer gain: here a deadbeat gain $H$ is used, whereas the innovation form employs the Kalman gain $K$. 
With $K$, the effect of the initial condition, $(A-KC)^\ell \hat X(0)$, generally does not vanish, so the representation is not strictly equivalent to an ARX model. 
For related work, see~\citep{Wang2022}. The matrices $(A,B,C,D,E,F,H,K)$ are introduced for exposition and need not be known.

\begin{assumption}\label{ass:iid_assum} The process $V$ $\in (\splx{\dimy})^\T$ is i.i.d with zero mean. During the data collection phase, $V^\text{d}(k)$ is independent of $Z^\text{d}(k)$ and $\Sinput^\text{d}(k)$ at time instant $k \in \I_{[0,T-1]}$, i.e., $\mean\left[V^\text{d}(k) \begin{bmatrix}
		Z^\text{d}(k)\\
	\Sinput^\text{d}(k)
	\end{bmatrix}^\top\right] = 0_{n_y\times (n_z+n_{\sinput})}$.
\end{assumption}
This assumption differs from assuming i.i.d.\ $\UnsD$, which would include i.i.d.\ process and measurement noise as a special case.
While Assumption~\ref{ass:iid_assum} is commonly adopted in ARX-based methods, it may be restrictive in practice (e.g., in building systems with temporal correlations or closed-loop operation). 
In such cases, it can be interpreted as an approximation after aggregating unmodeled effects. 
Despite this idealization, the proposed method shows good empirical performance in the smart-home case study (cf.~Sec.~\ref{sec:case}). 
Under this assumption, the residual disturbance can be consistently estimated, which in turn enables a distributionally consistent stochastic predictor.

\section{Consistent Data-driven Prediction}~\label{sec:data} Based on the dataset $\mcl D_T$ and its realization $D_T$, we define the stochastic regressor matrix
\[
\mcl S \doteq 
\begin{bmatrix}
\hankel_1(Z^\text{d}_{[0,T-1]})\\
\hankel_1(\Sinput^\text{d}_{[0,T-1]})
\end{bmatrix}, 
\quad 
S = \mcl S(\omega) = 
\begin{bmatrix}
\zdd\\
\udd
\end{bmatrix}.
\]
The realized data of \eqref{eq:ARX_stochastic} satisfy~\citep{Pan2022a}
\begin{equation}\label{eq:leftkern_ARX}
(\ydd - \edd)(I_T - S^\dagger S) = 0,
\end{equation}
where $I_T$ is the identity matrix and $(\cdot)^\dagger$ denotes the Moore--Penrose inverse. 
Since \eqref{eq:leftkern_ARX} admits infinitely many solutions for $\ue^\text{d}_{[0,T-1]}$, we select the least-squares solution
\[ \textstyle
\hat{\ue}^\text{d}_{[0,T-1]} = \argmin_{\ue^\text{d}_{[0,T-1]}} \|\ue^\text{d}_{[0,T-1]}\|^2 
\;\text{s.t.}\; \eqref{eq:leftkern_ARX},
\]
which admits the closed-form~\citep{Pan2022a}
\begin{equation}\label{eq:noise_estimation_LSE}
\Hankel_1(\hat{\ue}^\text{d}_{[0,T-1]}) = \ydd (I_T - S^\dagger S).
\end{equation}
While \eqref{eq:noise_estimation_LSE} coincides with a least-squares residual, the key difference lies in its role: 
the residual is interpreted as a stochastic disturbance process rather than a fitting error. 
Accordingly, we model $\hat \uE \in (\splx{\dimy})^{\mathbb{Z}}$ as an i.i.d.\ process with empirical distribution
\begin{equation}\label{eq:Vestimate}
\textstyle \hat \uE(k) \sim \mu_{\hat \uE}, \quad 
\mu_{\hat \uE} \doteq \frac{1}{T} \sum_{i=0}^{T-1} \delta_{\hat{\ue}^\text{d}(i)},
\end{equation}
where $\delta_{\hat{\ue}^\text{d}(i)}$ denotes the Dirac measure at $\hat{\ue}^\text{d}(i)$. 
This interpretation enables uncertainty propagation via the empirical distribution and underpins the consistency result in Theorem~\ref{thm:residual_consistency}.

Estimating the residual disturbance induces an ARX model that enables causal prediction, as shown next.

\begin{lemma}[Induced LTI representation~\citep{Pan2023d}]\label{pro:estimate_system}
Consider a realization $(z,\sinput,\ue,y)^\text{d}_{[0,T-1]}$ of \eqref{eq:ARX_stochastic}. 
For $\hat \ue^\text{d}_{[0,T-1]}$ satisfying \eqref{eq:noise_estimation_LSE}, define
$
\begin{bmatrix}
\widehat{\Xi} & \hat D
\end{bmatrix} \doteq \ydd S^\dagger.
$
Then $(z,\sinput,\hat \ue,y)^\text{d}_{[0,T-1]}$ is a trajectory of
\begin{equation}\label{eq:estimate_system}
\hat Y(k) = \widehat{\Xi} Z(k) + \hat D \Sinput(k) + \hat \uE(k), \quad Z(0)=Z_0.
\end{equation}
\end{lemma}

Next, we examine the convergence properties of the least-squares estimates of $\ue^\text{d}$ and the empirical estimation of $V_{[0,N-1]}$.
To this end, we have the following assumptions.
\begin{assumption}~\label{ass:PE}
	 $ \col(\sinput^\text{d},\hat \ue^\text{d})_{[0,T-1]}$ is  persistently exciting of order $N+\dimz $. 
\end{assumption}
With  residual disturbance $\hat \ue^\text{d}$ to be estimated, this condition can be verified a posteriori and, if necessary,
ensured by modifying the persistently exciting input $u^\text{d}$.
By Corollary 2 of \citep{Willems2005} and Lemma~\ref{pro:estimate_system}, Assumption~\ref{ass:PE} guarantees that  $S$ has full row rank, which implies
 $	S^\dagger = S^\top (SS^\top)^{-1},$ and $ SS^\dagger = I.$

\begin{assumption}\label{ass:convergence}
The  convergence  
$
\displaystyle \plim_{T \to \infty} \frac{1}{T}\mcl S \mcl S^\top = \tilde S
$ holds.
\end{assumption}
This assumption is typically satisfied for stable closed-loop systems driven by a persistently exciting input.

Note that $\hat v^\text{d}$ and $\mu_{\hat V}$ depend on the outcome $\omega \in \Omega$. 
Let $\hat V^\text{d}_{[0,T-1]}$ denote the random variable defined by 
$\hat V^\text{d}_{[0,T-1]}(\omega)\doteq \hat v^\text{d}_{[0,T-1]}$. Define $W_{V,\hat V} \in \spl$  whose realization represents the $2$-Wasserstein distance between the estimated and true disturbance distributions, i.e., $
W_{V,\hat V}(\omega)\doteq d_W(\mu_{\hat V}, \mu_V).$

\begin{theorem}[Consistent residual estimation]\label{thm:residual_consistency}
Suppose Assumptions~\ref{ass:iid_assum}--\ref{ass:convergence} hold. Then, $\hat V^\text{d}_{[0,T-1]}$ and $\mu_{\hat V}$ are consistent
\begin{align*}
\plim_{T\to\infty} \hat V^\text{d}_{[0,T-1]}(\mcl D_T) &= V^\text{d}_{[0,T-1]}, \quad
\plim_{T\to\infty} W_{V,\hat V}(\mcl D_T) = 0.
\end{align*}
\end{theorem}
\begin{pf}
	Substituting \eqref{eq:noise_estimation_LSE} into \eqref{eq:leftkern_ARX} shows that 	$\hat{\ue}^\text{d}_{[0,T-1]}$ satisfies the same left-kernel condition with $y^\text{d}_{[0,T-1]}$ replaced by $v^\text{d}_{[0,T-1]}$, which implies 
\[
\plim_{T\to\infty}\Hankel_1(V^\text{d}_{[0,T-1]}) - 	\Hankel_1(\hat V^\text{d}_{[0,T-1]}) {=} \plim_{T\to\infty} 	\Hankel_1(V^\text{d}_{[0,T-1]}) \mcl S^\dagger \mcl S  
\]	
Following the standard results of \cite[p.~205]{Ljung1999}, denoted by $(L)$,
and \cite[p.~97]{Norton1988}, denoted by $(N)$, and exploiting Assumptions~\ref{ass:iid_assum}--\ref{ass:convergence}, we obtain
	 \begin{align*}
	 	& \plim_{T\to\infty} 	\Hankel_1(V^\text{d}_{[0,T-1]}) \mcl S^\dagger \mcl S {=}\plim_{T\to\infty}\Hankel_1(V^\text{d}_{[0,T-1]})  \mcl S^\top( \mcl S \mcl S^\top)^{-1}\mcl S \\
	 &\stackrel{(N)}{=}\plim_{T\to\infty}\frac{1}{T} \Hankel_1(V^\text{d}_{[0,T-1]})  \mcl S^\top \plim_{T\to\infty}\left(\frac{1}{T} \mcl S \mcl S^\top\right)^{-1} \mcl S \\
 	& \stackrel{(L)}{=}\mean\left[V^\text{d}(k) \begin{bmatrix}
 		Z^\text{d}(k)\\
 		\tilde U^\text{d}(k)
 	\end{bmatrix}^\top\right]  \tilde S ^{-1} \mcl S  \stackrel{\text{Assump. \ref{ass:iid_assum}, \ref{ass:convergence}}}{=}0.
	 \end{align*}

    For distribution, let $\mu_{\tilde V}$ denote the empirical measure of $V^\text{d}$, $\mu_{\tilde \uE} = \frac{1}{T} \sum_{i=0}^{T-1} \delta_{\ue^\text{d}\inst{i}}$. 
       By the triangular inequality,
	 we have 
$
	 W_{\hat V,  V}\leq W_{\hat V, \tilde V}+ W_{\tilde V,  V}.
$ 
The first term vanishes by consistency of $\hat V^\text{d}$ as proven in the first part, while the second converges to zero by classical empirical-measure convergence. \hfill $\square$
\end{pf}

We now proceed to construct a data-driven stochastic predictor.  To this end, consider the partition of the Hankel matrix 
$
 \mathcal H_{N+\ell} \left(\sinput^\text{d}_{[0,T-1]}\right) \doteq\begin{bmatrix}
 	\hankel_{\sinput,\text{p}}\\
 	\hankel_{\sinput,\text{f}}
 \end{bmatrix},
$
 where $\hankel_{\sinput,\text{p}}$ contains the first $\ell$ block rows, and $\hankel_{\sinput,\text{f}}$  the rest. Similarly, we define  $\hankel_{y,\text{p}}$, $\hankel_{y,\text{f}}$,  $\hankel_{\hat v,\text{f}}$, and $ \hankel_\text{p} \doteq \begin{bmatrix}
 		\hankel_{\sinput,\text{p}}\\
 			\hankel_{y,\text{p}}
 \end{bmatrix}$.
\begin{corollary}
\label{cor:FL_stoch_ext_noise}
Let Assumption~\ref{ass:PE} hold. 
			$(\Sinput,  \hat \uE,  \hat Y)_{[-\ell,N-1]} $ is a trajectory of the estimated system~\eqref{eq:estimate_system} with initial condition $ Z(0) \doteq \begin{bmatrix}
				\Sinput_{[-\ell,-1]}\\
				\hat Y_{[-\ell,-1]}
			\end{bmatrix}$  iff there is $G\in \splx{T-N+1}$ such that
			\begin{subequations}\label{eq:Hankel}
				\begin{align}
					\hankel_\text{p} G&= Z(0), \quad
					\hankel_{\sinput,\text{f}} G = \Sinput_{[0,N-1]},\label{eq:U_prediction}\\
					\hankel_{y,\text{f}} G&= \hat Y_{[0,N-1]},					\quad	
					\hankel_{\hat v,\text{f}} G = \hat V_{[0,N-1]}.	\label{eq:residual_cons} 				
				\end{align}
			\end{subequations}
\end{corollary}
This corollary follows directly from Lemma 4 of \citep{Pan2023d}, Lemma~\ref{pro:estimate_system}, and the original deterministic fundamental lemma in \cite{Willems2005}. Moreover,  adapted from the lines in Theorem 1 of \citep{Fiedler2021}, \eqref{eq:Hankel} admits the unique solution in closed form 
\begin{equation}~\label{eq:causal_predictor}
	\hat Y_{[0,N-1]} = \hankel_{y,\text{f}}\begin{bmatrix}
		\hankel_{\text{p}} \\
		\hankel_{\sinput,\text{f}} \\
	\hankel_{	\hat  \ue,\text{f}}
	\end{bmatrix}^\dagger \begin{bmatrix}
		Z(0)\\
		\Sinput_{[0,N-1]}\\
		\hat{ \uE}_{[0,N-1]}
	\end{bmatrix}.
\end{equation}
Let  $W_{Y,\hat Y} \in \spl$ be defined as 
$W_{Y,\hat Y} (\omega)\doteq d_W(\mu_{\hat Y_{[0,N-1]}},\mu_{ Y_{[0,N-1]}})$ for $\omega \in \Omega$. The consistent estimation of the residual disturbance in Theorem~\ref{thm:residual_consistency} implies the following result. 
\begin{corollary}
 Suppose Assumptions~\ref{ass:iid_assum}--\ref{ass:convergence} hold. Then, 	$\hat Y_{[0,N-1]} $ by \eqref{eq:causal_predictor} is a distributionally consistent estimate, i.e., $
\plim_{T\to \infty}   W_{Y,\hat Y}(\mcl D_T) {=} 0.$
 \end{corollary}
Ignoring residual disturbances,
 the classical subspace predictor uses the least-squares solution of \eqref{eq:U_prediction} \citep{Favoreel1999,Fiedler2021}
	\begin{equation*}\label{eq:Yf_opt}
		\hat Y_{[0,N-1]}^\text{SP} = \hankel_{y,f} G^\star,\, G^\star =	\argmin_{G} \Lnorm{G}^2 \text{ s.t. \eqref{eq:U_prediction}},
	\end{equation*}
whose closed-form solution reads
\begin{equation}\label{eq:subspace}
\hat Y_{[0,N-1]}^\text{SP} = \hankel_{y,\text{f}}\begin{bmatrix}
	\hankel_{\text{p}} \\
	\hankel_{\sinput,\text{f}} 
\end{bmatrix}^\dagger \begin{bmatrix}
	Z(0)\\
	\Sinput_{[0,N-1]}
\end{bmatrix}.
\end{equation}
The subspace predictor is \emph{acausal}, since $\hat Y^\text{SP}(k)$ may depend on future inputs $\Sinput(i \ge k)$ for $k \in \I_{[0,N-1]}$. 
Although causality can be enforced by imposing a Toeplitz (block lower-triangular) structure, the resulting predictor remains deterministic for deterministic future inputs. 
In contrast, the proposed predictor~\eqref{eq:causal_predictor}, which explicitly accounts for the residual disturbance, corresponds to the output response of the estimated system~\eqref{eq:estimate_system} and is therefore \emph{causal}, while also capturing uncertainty arising from unmeasured or unmodeled effects.

Regarding statistical properties,~\cite{Ljung1996} show that the subspace predictor is mean-consistent. 
~\cite{Fiedler2023} establish distributional consistency using empirical covariance estimates from multiple trajectories, which limits applicability to multi-trajectory settings rather than a single long-run experiment. 

\section{Tractable Uncertainty Propagation}\label{sec:UQ}
In this section, we present tractable uncertainty propagation via PCE, which represents $L^2$ random variables using orthogonal polynomials~\citep{Wiener1938,Sullivan2015}. Let  $ \xi \in \splx{n_{\xi}}$ denote a vector of independent random variables  generating  $\mcl{F} \doteq \sigma(\xi)$. 
Then, every $M \in \spl$  admits an $L^2$-convergent expansion
	\begin{equation*}\label{eq:PCE_def}
		M = \sum_{j=0}^{\infty}\pcecoe{m}{j} \phi^j( \xi)  \quad\text{with}\quad \pcecoe{m}{j} = \frac{\mean[M \phi^j(\xi)]}{\Lnorm{\phi^j(\xi)}^2 },
	\end{equation*}
	where $\pcecoe{m}{j} \in \R$ denotes the $j$-th PCE coefficient and ${\phi^j(\xi)}$ is an orthogonal basis.

For a  vector-valued $V\in \splx{\dimv}$, coefficients are collected as $\pce{v}^j \in \R^{\dimv}$.
Truncating the series to a finite number of terms may introduce  errors~\citep{Field2004,Muehlpfordt2018};  if exact,
$
V- \sum_{j=0}^{L-1} \pcecoe{v}{j} \phi^j(\xi)  =0. 
$
Then, $V\in \splx{\dimv}$ is said to \emph{admit an exact PCE with $L$ terms} in $\{\phi^j(\xi)\}_{j\in\N}$.
For compactness, we define 
\begin{align*}
\pcecoe{v}{[0,L-1]} &\doteq \left[\pcecoe{v}{0\top},\pcecoe{v}{1\top},\dots,\pcecoe{v}{L-1\top}\right]^\top \in \R^{\dimv L},\\
 \hankel_1\left(\pcecoe{v}{[0,L-1]}\right) &\doteq \left[\pcecoe{v}{0},\pcecoe{v}{1},\cdots, \pcecoe{v}{L-1}\right] \in \R^{\dimv \times L}.
 \end{align*}

The construction of exact finite-dimensional PCEs depends on the chosen basis. In particular, an exact PCE can be obtained from the first two moments.
For example, any  scalar  $M \in$  $\spl$ admits the affine expansion 
\[
M = \mean[M] + \sqrt{\mu_2(M)} \xi, \quad   \xi = (M - \mean[M])(\mu_2(M))^{-\frac{1}{2}},
\]
which is a  two-term PCE in the basis $\{1, \xi\}$. This extends  to vector-valued random variables.
\begin{lemma}[Affine PCEs~{\citep[Lemma 2]{Pan2023d}}] \label{lem:finite_PCE}~\\
	Let $V \in\splx{\dimv}$ have mean $\mean[V] \in \R^{\dimv}$ and covariance $\Sigma[V] \succeq 0$. For any  $\pce V  \in \R^{\dimv \times n_\xi}$ satisfying
	$
	\pce V \pce V^{\top} =\Sigma[V],
	$
	there exists  $\xi \in \splx{n_\xi}$  with   $\mean[\xi]=0\text{ and } \Sigma[\xi] = I_{n_\xi}$ such that
	\[
	V= \mean[V] + \pce V \xi= \textstyle \sum_{j=0}^{n_\xi} \pcecoe{v}{j} \phi^j(\xi),
	\] 	
where $\{\phi^j(\xi)\}_{j=0}^{n_\xi} \doteq	\left\{1,\{\xi_i\}_{i=1}^{n_\xi}\right\} $,  $\pcecoe{v}{0} \doteq \mean[V]$, and $ \hankel_1\left(\pcecoe{v}{[1,\dimv]}\right) \doteq \pce V $.
\end{lemma}

Applying Lemma~\ref{lem:finite_PCE} to the estimated
$\hat \uE \inst{k}$, we have
\begin{align*}
	\hat \uE \inst{k} &= \hat m_\ue + \hat{\Sigma}_\ue^{\frac{1}{2}} \xi \inst k, \quad
	\xi \inst k   \textstyle \sim  \frac{1}{T} \sum_{i=0}^{T-1} \delta_{\xi^{\text d} \inst{i}},
\end{align*}
with $\hat m_v$ the empirical mean, $\hat \Sigma_v$ the empirical covariance, and $\xi^{\text d}(i) = \hat \Sigma_v^{-\frac{1}{2}}\big(\hat v^{\text d}(i) - \hat m_v\big)$. Here, $\xi \inst k$ follows a normalized empirical distribution with zero mean and unit covariance.
Hence, $\hat \uE\inst{k}$ admits an exact PCE in the basis $\{1,  \xi_1\inst{k},\cdots,\xi_{n_y}\inst{k} \}$,
where $\xi_i\inst{k}$, $i \in \I_{[1,\dimy]}$ are elements of $\xi\inst{k}$. 
Assume $\Sinput_{[0,N-1]}$ and $Z_0$ admits exact PCEs in the basis $\{\bar \phi\}_{j=0}^{\bar L-1}$. Then, a joint basis for $\Sinput_{[0,N-1]}$,  $Z_0$, and $\hat{\uE}_{[0,N-1]}$ is given by
\begin{equation}\label{eq:union_basis}
	\{\phi^j\}_{j=0}^{L-1} =  \{\bar \phi^j\}_{j=0}^{\bar L-1} \bigcup  
	\left\{  \xi_1\inst{k},\cdots,\xi_{n_y}\inst{k} \right\}_{k=0}^{N-1}.
\end{equation}
with $L = \bar L + N \dimy$.

Although higher-degree bases can be constructed via Gram–Schmidt~\citep{Witteveen2006}, this first-order basis suffices to exactly represent all affine transformations of $\hat V_{[0,N-1]}$, $\Sinput_{[0,N-1]}$, and $Z_0$~\citep{Muehlpfordt2018}. Consequently, both predictors~\eqref{eq:causal_predictor} and~\eqref{eq:subspace} admit exact PCE representations in this basis.

Consider any scalar component $M \in \spl$ of the predicted output $\hat Y_{[0,N-1]}$. We quantify uncertainty via a confidence interval $I_\gamma \subseteq \R$ such that
\[
\mathbb{P}\!\left[M \in I_\gamma\right] \ge \gamma, \quad 0<\gamma<1.
\]

Using the mean $\mean[M]$ and variance $\mu_2(M)$, Chebyshev’s inequality~\citep{Tchebychef1867} yields
\begin{subequations}\label{eq:chebyshev}
\begin{align}
\mathbb{P}&\!\left[\big|M-\mean[M]\big| \le \alpha(\gamma)\sqrt{\mu_2(M)}\right] \ge \gamma, \\
\alpha(\gamma)&=
\begin{cases}
(1-\gamma)^{-\frac{1}{2}}, & M \in \spl,\\
F^{-1}_{\mcl N}\!\left(\tfrac{1+\gamma}{2}\right), & M \sim \mcl N(\mean[M],\mu_2(M)),
\end{cases}
\end{align}
\end{subequations}
where $F^{-1}_{\mcl N}$ denotes the inverse CDF of the standard normal distribution. The bound for general $L^2$ variables is conservative, while the Gaussian bound may underestimate tail probabilities. Higher-order moments yield a refinement~\citep{Buot2006}.

\begin{lemma}[Higher-order Chebyshev inequalities]\label{lem:Chebyshev_high_order}
Let $M \in \spl$ have finite $2n$-th moment $\mu_{2n}(M)$. Then, for any $0<\gamma<1$,
\begin{equation}
\mathbb{P}\!\left[\big|M-\mean[M]\big| \le \big((1-\gamma)^{-1}\mu_{2n}(M)\big)^{\frac{1}{2n}}\right] \ge \gamma.
\end{equation}
\end{lemma}

The corresponding confidence interval is
\[
I_{2n,\gamma} = \big[\mean[M]- r_{2n,\gamma},\, \mean[M] + r_{2n,\gamma}\big], 
r_{2n,\gamma} = \left(\frac{\mu_{2n}(M)}{1-\gamma}\right)^{\frac{1}{2n}},
\]
which is tighter than the Chebyshev interval $I_{2,\gamma}$ if
\[
\mu_{2n}(M) \le (1-\gamma)^{1-n}\big(\mu_2(M)\big)^n.
\]
For $n=2$ and $\gamma=0.9$, this requires the kurtosis $\kappa = \mu_4/\mu_2^2 \le 10$. Since a normal distribution has $\kappa=3$, this condition excludes extreme heavy tails while allowing moderate non-Gaussianity.

If $M$ admits an exact PCE, higher-order moments can be computed efficiently.  For $M = \sum_{j_1=1}^{L-1} \pcecoe{m}{j}\phi^{j}( \xi)$, 
\begin{multline}
	\mu_{2n} (M)= \textstyle
	 \sum_{j_1=1}^{L-1} \sum_{j_2=1}^{L-1}\cdots \sum_{j_{2n}=1}^{L-1} \pcecoe{m}{j_1} \pcecoe{m}{j_2} \cdots \pcecoe{m}{j_{2n}} \\ \cdot \mean\left[\phi^{j_1}( \xi)\phi^{j_2}( \xi)\cdots\phi^{j_{2n}}( \xi)\right], \nonumber  
	  \end{multline}
      where expectations of basis products can be precomputed offline~\citep{Lefebvre2020}. If the basis functions are independent distributed as in~\eqref{eq:union_basis}, then
\[
\mathbb{E} \left[(\phi^i( \xi))^{l_1} (\phi^i( \xi))^{l_2}\right] = \mathbb{E} \left[(\phi^i( \xi))^{l_1}\right] \mathbb{E} \left[(\phi^i( \xi))^{l_2}\right].
\]
and $\mathbb{E}[\phi^i(\xi)] = 0$ for $i \ge 1$. 
 Hence, the fourth-order moment simplifies to
\begin{multline}\label{eq:fourth_moment_PCE}
		\mu_4(M)  
		= \textstyle \sum_{j=1}^{L-1}  (\pcecoe{m}{j})^4 \mu_4(\phi^{j}( \xi)) \\
		 \textstyle+  \sum_{i=1}^{L-1} \sum_{j=i+1}^{L-1} 6 (\pcecoe{m}{i})^2 (\pcecoe{m}{j})^2 \mu_2(\phi^{j}( \xi)) \mu_2(\phi^{i}( \xi)).
	\end{multline}

\section{A Smart-Home Case Study}\label{sec:case}We evaluate the proposed stochastic prediction method on a real-world smart-home dataset from a residential building in Norway operated by NTNU~\citep{Reinhardt2024}. The dataset and preprocessing code are provided in~\citep{Reinhardt2024,Kaupmann2020}.

The building comprises four thermal zones with distinct heat capacities and losses, each equipped with a temperature sensor and a heat-pump-based heating device. The control inputs include the ON/OFF status $O_i \in \{0,1\}$, target temperature $T_i^{\text{set}} \in \I_{[16,31]}~^\circ$C, and fan speed $F_i \in \I_{[1,7]}$, forming
\begin{equation}\label{eq:input}
u \doteq \col(O,T^{\text{set}},F)_{[1,4]} \in \R^{12}.
\end{equation}
The outputs are the zone temperatures
\begin{equation}\label{eq:output}
y \doteq [T_1,\ldots,T_4]^\top \in \R^4.
\end{equation}
Measured disturbances include solar azimuth $\theta_a$, solar zenith $\theta_z$, outdoor temperature $T^{\text{out}}$, wind speed $v_w$, cloud opacity $\tau_c$, and global irradiance $I_G$:
\begin{equation}\label{eq:measured_disturbance}
w^{\text{m}} \doteq [\theta_a, \theta_z, T^{\text{out}}, v_w, \tau_c, I_G]^\top,
\end{equation}
while $w^{\text{u}}$ denotes unmeasured disturbances. Figure~\ref{fig:building_sketch} illustrates the relation between $u$, $y$, $w^{\text{m}}$, and $w^{\text{u}}$.

The dataset spans May 11 to December 29, 2021, with a 15-minute sampling interval. 
At each prediction time $t$, the past 30 days ($T=2880$) are used to estimate the residual disturbance $\hat{\ue}^{\text{d}}$ and predict $\hat{\uE}$ over a 24-hour horizon ($N=96$). 
Predictions are generated every 6 hours from Nov.~8 to Dec.~27, 2021, resulting in $N_s=200$ scenarios. 
To isolate the effect of unstructured disturbances, future inputs and measured disturbances are treated as deterministic. 
Accordingly, $\bar\phi = \{1\}$ in~\eqref{eq:union_basis}, and the fourth-order moments of $\hat{Y}(k\,|\,t)$ follow~\eqref{eq:fourth_moment_PCE}.

\begin{figure}[tb]
	\centering
	\includegraphics[width=.5\linewidth]{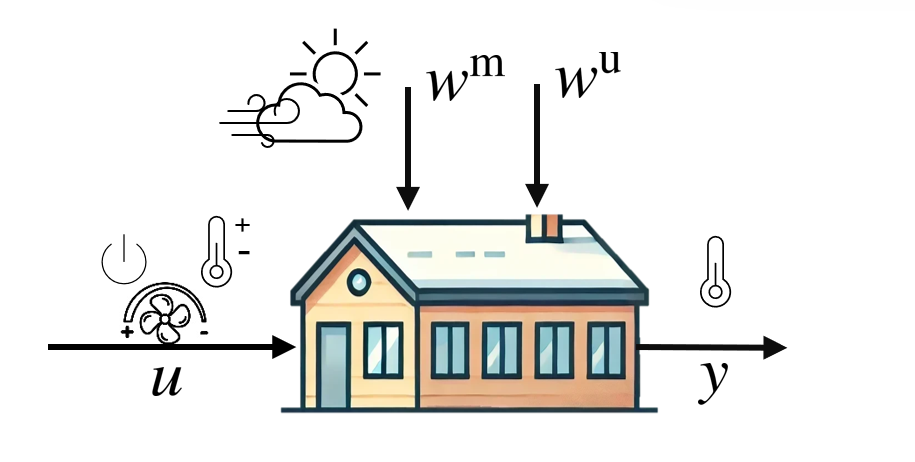}
	\caption{Conceptual block diagram of the residential building with input $u$~\eqref{eq:input}, output $y$~\eqref{eq:output}, measured disturbance $w^{\text{m}}$~\eqref{eq:measured_disturbance}, and unmeasured disturbance $w^{\text{u}}$.}\label{fig:building_sketch}
\end{figure}

\begin{table}[t!]
\setlength{\tabcolsep}{1pt} 
	\caption{Comparison of the prediction results.}
	\label{tab:NTNU_building_lambda}
	\centering
		\begin{adjustbox}{width=\columnwidth}
		\begin{tabular}{@{}cccccccc@{}}
			\toprule
			\multirow{2}{*}{Approach}&  	\multirow{2}{*}{$\overline{	\text{RMSE}}$ }   &  \multicolumn{2}{c}{Cheb. $\mu_2$} &  \multicolumn{2}{c}{Cheb. $\mu_4$} &\multicolumn{2}{c}{Gauss.}\\ 
			\cmidrule(lr){3-4}   \cmidrule(lr){5-6}   \cmidrule(lr){7-8} 
			& 	& $\bar c$ $[-]$ & $\bar{r}$ $\SI{}{[^\circ C]}$	& $\bar c$ $[-]$ & $\bar{r}$ $\SI{}{[^\circ C]}$	& $\bar c$ $[-]$ & $\bar{r}$ $\SI{}{[^\circ C]}$\\
			\midrule
			\eqref{eq:subspace}, $w^\text{s} = \emptyset$ & 1.068 &	0 &	0	  &0 &	0 &0 &0	 \\
			\eqref{eq:causal_predictor}, $w^\text{s} = \emptyset$ &  0.859 &	97.29$\%$ &	2.36 &	94.34$\%$ &	1.77 &	85.96$\%$ &	1.23\\ 
			\eqref{eq:causal_predictor}, $w^\text{s}=[\theta_a,\theta_z]^\top$  & 	0.802 &	96.90$\%$&	2.07 &	93.40$\%$ &	1.56 &	84.02$\%$ &	1.08\\
			\eqref{eq:causal_predictor}, $w^\text{s}=w^{\text m}$ & 	0.785 &	95.97$\%$ &	1.78 &	91.11$\%$ &	1.35&	78.62$\%$ &	0.93\\   
			\bottomrule
		\end{tabular}
		\end{adjustbox}
	\end{table}
	Let $\pcecoe{\hat y}{j}(k\,|\,t)$ denote the PCE coefficients of the predicted output, where $\pcecoe{\hat y}{0}(k\,|\,t)$ is the predicted mean. 
The mean prediction accuracy is evaluated by the average root mean square error (RMSE)
\begin{equation}
\overline{\text{RMSE}} \doteq \sqrt{\frac{1}{NN_s} \sum_{t=0}^{N_s-1} \sum_{k=0}^{N-1} 
\left\| \pcecoe{\hat y}{0}(k\,|\,t) - y(t+k) \right\|^2 }.
\end{equation}
We access stochastic accuracy via the average $90\%$ coverage 
\begin{equation*} \bar c \doteq \sum_{t=0}^{N_s-1}\sum_{i =1}^{\dimy} \sum_{k = 0}^{N-1} \frac{ I\left(\left|y_i\inst{t + k} - \pcecoe{\hat y}{0}_i\inst{k \,|\,t } \right| \leq r_{i}\left(k\,|\,t\right)\right) }{NN_s \dimy}, \end{equation*}
where $I(\cdot)$ is the indicator function. The radius $r_i(k\,|\,t)$ is chosen from
\begin{equation}\label{eq:radius_calculation}
r_i(k\,|\,t) \approx 
\begin{cases}
3.16 \sqrt{\mu_2}, & \text{(Chebyshev $\mu_2$)},\\
1.78 \sqrt[4]{\mu_4}, & \text{(Chebyshev $\mu_4$)},\\
1.64 \sqrt{\mu_2}, & \text{(Gaussian)}.
\end{cases}
\end{equation}
The average interval radius $\bar r$ is defined analogously. Coverage and interval width jointly reflect the trade-off between reliability (coverage) and conservatism (interval size).

Table~\ref{tab:NTNU_building_lambda} summarizes the prediction results. 
Among the variants in~\eqref{eq:radius_calculation}, the fourth-order Cheb.\ approach yields tighter confidence bounds than standard Cheb.\ and higher coverage than the Gauss.\ assumption, providing a good balance between coverage and interval width. 
The subspace predictor~\eqref{eq:subspace}, being deterministic, fails to capture uncertainty from unmeasured disturbances.

We further assess the impact of structured disturbances: 
(i) none; 
(ii) solar angles $w^\text{s} = [\theta_a,\theta_z]^\top$; 
(iii) all measured disturbances $w^\text{s} = w^\text{m}$. 
Including solar angles improves $\overline{\text{RMSE}}$ by about $0.05^\circ$C and reduces the average interval radius by $0.2$--$0.3^\circ$C. 
These angles are deterministically computable from time and location. 
In contrast, adding weather variables ($T^{\text{out}}, v_w, \tau_c, I_G$) yields only marginal improvement ($\approx 0.02^\circ$C in $\overline{\text{RMSE}}$), even with perfect forecasts.

Figure~\ref{fig:Nov_27_response} shows true and predicted $T_1$ for a representative room--day instance. 
For this case, the RMSE of \eqref{eq:subspace} is $0.47^\circ$C (no structured disturbance), while \eqref{eq:causal_predictor} achieves $0.42^\circ$C, $0.30^\circ$C, and $0.34^\circ$C for the three cases above. 
Overall, incorporating residual disturbances reduces prediction error and enables stochastic predictions with confidence intervals. 
Including solar angles significantly improves accuracy and tightens intervals, whereas additional weather variables provide limited benefit. 
Although trajectories vary across the dataset, the overall performance trends remain consistent.

\begin{figure}[t]
	\centering
    \includegraphics[width=.9\columnwidth]{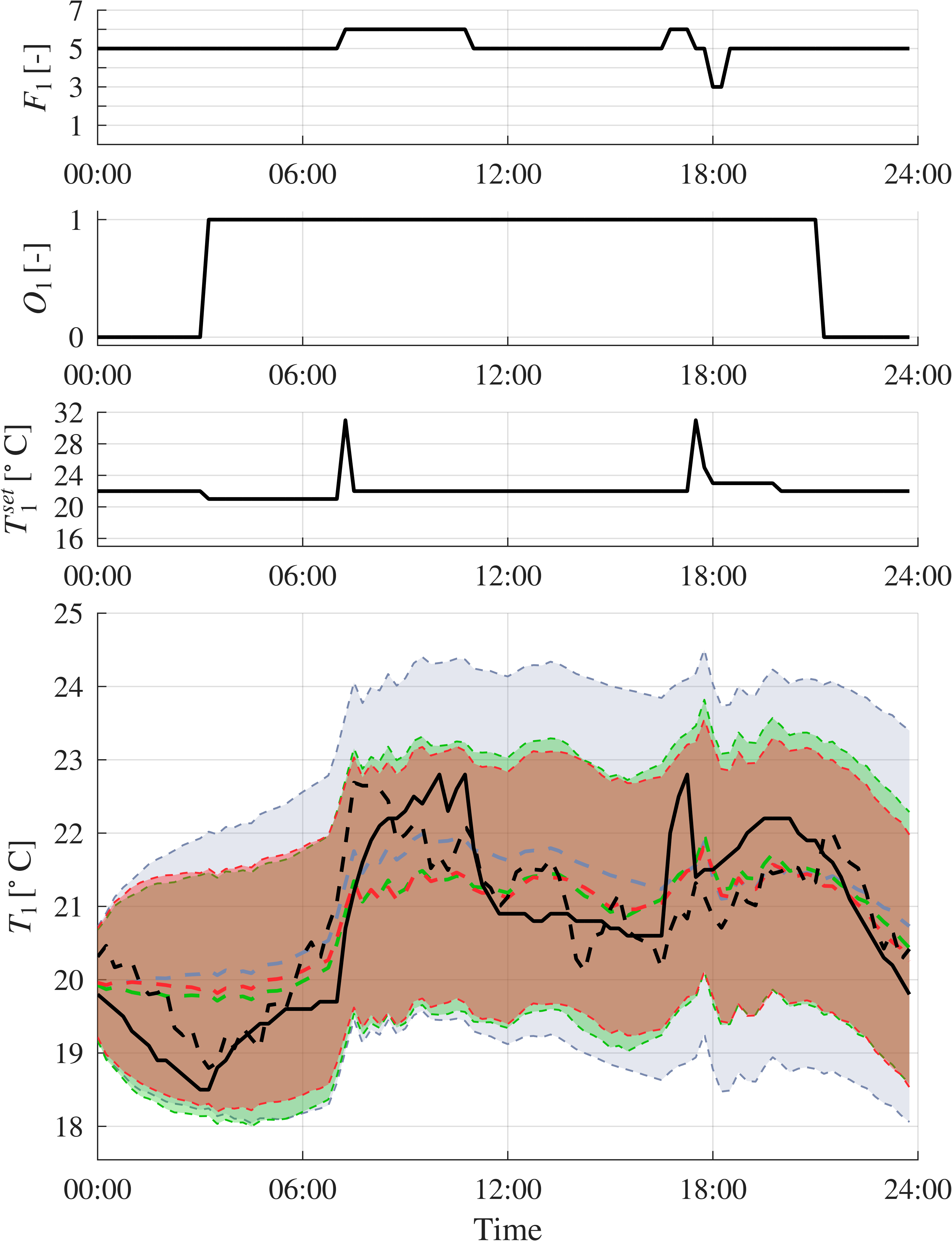} 
\caption{Comparison of true and predicted thermal responses on November 17, 2021. Black solid lines show measured inputs and outputs (ground truth). Blue, green, and red dashed lines denote predicted mean trajectories with $90\%$ confidence intervals (based on fourth-order moments) for the cases $w^\text{s} = \emptyset$, $w^\text{s} = [\theta_a, \theta_z]^\top$, and $w^\text{s} = w^\text{m}$, respectively. The black dashed line shows the subspace predictor~\eqref{eq:subspace} without confidence intervals for $w^\text{s} = \emptyset$.} \label{fig:Nov_27_response} 
\end{figure}

\section{Conclusion}\label{sec:conclusions}
This paper introduced a data-driven framework for stochastic prediction that  addresses unmeasured and unmodeled disturbances through a residual disturbance process. By combining behavioral system theory with PCE, the proposed approach enables tractable uncertainty propagation directly from input–output data.
Theoretical results established the consistency of residual disturbance estimation and the distributional convergence of the proposed predictor. Application to smart-home data demonstrates the effectiveness of the proposed method.
 Among the uncertainty-bounding schemes, the fourth-order Chebyshev approach achieved the best balance between coverage probability and tightness of confidence intervals.
Future work includes relaxing the i.i.d. assumption on residuals and extending the framework to nonlinear systems via kernel embeddings. Moreover, also application to chemical process systems seems promising.

\begin{ack}
	The authors gratefully acknowledge valuable discussions with Alexander Engelmann and  Michael Kaupmann.
\end{ack}

\section*{DECLARATION OF GENERATIVE AI IN THE WRITING PROCESS}
During the preparation of this work the authors used ChatGPT to assist with language editing and text refinement. After using this tool/service, the authors reviewed and edited the content as needed and take full responsibility for the content of the publication.

\begingroup
\setlength{\bibsep}{2pt}        
\renewcommand*{\bibfont}{\small}

\bibliography{ALL_GP,Publications}

\begin{thebibliography}{28}
\providecommand{\natexlab}[1]{#1}
\providecommand{\url}[1]{\texttt{#1}}
\providecommand{\urlprefix}{URL }
\expandafter\ifx\csname urlstyle\endcsname\relax
  \providecommand{\doi}[1]{doi:\discretionary{}{}{}#1}\else
  \providecommand{\doi}{doi:\discretionary{}{}{}\begingroup
  \urlstyle{rm}\Url}\fi

\bibitem[{Breschi et~al.(2023)Breschi, Chiuso, and Formentin}]{Breschi2023}
Breschi, V., Chiuso, A., and Formentin, S. (2023).
\newblock Data-driven predictive control in a stochastic setting: A unified
  framework.
\newblock \emph{Automatica}, 152, 110961.

\bibitem[{Buot(2006)}]{Buot2006}
Buot, M. (2006).
\newblock Probability and computing: randomized algorithms and probabilistic
  analysis.

\bibitem[{Chebyshev(1867)}]{Tchebychef1867}
Chebyshev, P.L. (1867).
\newblock Des valeurs moyennes.
\newblock \emph{Journal de Math\'ematiques Pures et Appliqu\'ees}, 2(12),
  177--184.

\bibitem[{Coulson et~al.(2019)Coulson, Lygeros, and D{\"o}rfler}]{Coulson2019}
Coulson, J., Lygeros, J., and D{\"o}rfler, F. (2019).
\newblock Data-enabled predictive control: In the shallows of the {DeePC}.
\newblock In \emph{2019 18th European Control Conference (ECC)}, 307--312.
  IEEE.

\bibitem[{Drgoňa et~al.(2020)Drgoňa, Arroyo, Cupeiro~Figueroa, Blum, Arendt,
  Kim, Ollé, Oravec, Wetter, Vrabie, and Helsen}]{Drgona2020}
Drgoňa, J., Arroyo, J., Cupeiro~Figueroa, I., Blum, D., Arendt, K., Kim, D.,
  Ollé, E.P., Oravec, J., Wetter, M., Vrabie, D.L., and Helsen, L. (2020).
\newblock All you need to know about model predictive control for buildings.
\newblock \emph{Annu. Rev. Control}, 50, 190--232.

\bibitem[{Favoreel et~al.(1999)Favoreel, De~Moor, and Gevers}]{Favoreel1999}
Favoreel, W., De~Moor, B., and Gevers, M. (1999).
\newblock {SPC: Subspace predictive control}.
\newblock \emph{IFAC Proc. Vol.}, 32(2), 4004--4009.

\bibitem[{Fiedler and Lucia(2021)}]{Fiedler2021}
Fiedler, F. and Lucia, S. (2021).
\newblock On the relationship between data-enabled predictive control and
  subspace predictive control.
\newblock In \emph{2021 European Control Conference (ECC)}, 222--229.

\bibitem[{Fiedler and Lucia(2023)}]{Fiedler2023}
Fiedler, F. and Lucia, S. (2023).
\newblock Probabilistic multi-step identification with implicit state
  estimation for stochastic {MPC}.
\newblock \emph{IEEE Access}, 11, 117018--117029.

\bibitem[{Field and Grigoriu(2004)}]{Field2004}
Field, R.V. and Grigoriu, M. (2004).
\newblock On the accuracy of the polynomial chaos approximation.
\newblock \emph{Probabilist. Eng. Mech.}, 19(1-2), 65--80.

\bibitem[{Huang et~al.(2023)Huang, Zhen, Lygeros, and D{\"o}rfler}]{Huang2023}
Huang, L., Zhen, J., Lygeros, J., and D{\"o}rfler, F. (2023).
\newblock Robust data-enabled predictive control: Tractable formulations and
  performance guarantees.
\newblock \emph{IEEE Trans. Autom. Control}, 68(5), 3163--3170.

\bibitem[{Kaupmann(2023)}]{Kaupmann2020}
Kaupmann, M. (2023).
\newblock \emph{Data-driven modeling for building control}.
\newblock Master's thesis, Institute of Energy Systems, Energy Efficiency and
  Energy Economics, TU Dortmund University.

\bibitem[{Lefebvre(2020)}]{Lefebvre2020}
Lefebvre, T. (2020).
\newblock On moment estimation from polynomial chaos expansion models.
\newblock \emph{IEEE Control Syst. Lett.}, 5(5), 1519--1524.

\bibitem[{Ljung(1999)}]{Ljung1999}
Ljung, L. (1999).
\newblock \emph{{System Identification: Theory for the User}}.
\newblock Prentice Hall.

\bibitem[{Ljung and McKelvey(1996)}]{Ljung1996}
Ljung, L. and McKelvey, T. (1996).
\newblock Subspace identification from closed loop data.
\newblock \emph{Signal Process.}, 52(2), 209--215.

\bibitem[{M{\"u}hlpfordt et~al.(2017)M{\"u}hlpfordt, Findeisen, Hagenmeyer, and
  Faulwasser}]{Muehlpfordt2018}
M{\"u}hlpfordt, T., Findeisen, R., Hagenmeyer, V., and Faulwasser, T. (2017).
\newblock Comments on quantifying truncation errors for polynomial chaos
  expansions.
\newblock \emph{IEEE Control Syst. Lett.}, 2(1), 169--174.

\bibitem[{Norton(1988)}]{Norton1988}
Norton, J.P. (1988).
\newblock \emph{An introduction to identification}.
\newblock Acad. Press, London [u.a.].

\bibitem[{Pan(2025)}]{Pan2025D}
Pan, G. (2025).
\newblock \emph{Data-Driven Control of Stochastic Systems: Representation,
  Prediction, and Optimal Control}.
\newblock Ph.D. thesis, Hamburg University of Technology.

\bibitem[{Pan et~al.(2023)Pan, Ou, and Faulwasser}]{Pan2022a}
Pan, G., Ou, R., and Faulwasser, T. (2023).
\newblock On a stochastic fundamental lemma and its use for data-driven optimal
  control.
\newblock \emph{IEEE Trans. Autom. Control}, 68(10), 5922--5937.

\bibitem[{Pan et~al.(2025)Pan, Ou, and Faulwasser}]{Pan2023d}
Pan, G., Ou, R., and Faulwasser, T. (2025).
\newblock On data-driven stochastic output-feedback predictive control.
\newblock \emph{IEEE Trans. Autom. Control}, 70(5), 2948--2962.

\bibitem[{Reinhardt et~al.(2025)Reinhardt, Cai, and Gros}]{Reinhardt2024}
Reinhardt, D., Cai, W., and Gros, S. (2025).
\newblock \emph{Data-driven domestic flexible demand: observations from
  experiments in cold climate}, 691--728.
\newblock Springer Nature Switzerland.

\bibitem[{Sullivan(2015)}]{Sullivan2015}
Sullivan, T.J. (2015).
\newblock \emph{{Introduction to Uncertainty Quantification}}, volume~63.
\newblock Springer.

\bibitem[{{van Overschee} and De~Moor(2012)}]{VanOverschee2012}
{van Overschee}, P. and De~Moor, B. (2012).
\newblock \emph{{Subspace identification for linear systems:
  Theory-Implementation-Applications}}.
\newblock Springer Science \& Business Media.

\bibitem[{Wang et~al.(2025)Wang, You, Huang, and Shang}]{Wang2022}
Wang, Y., You, K., Huang, D., and Shang, C. (2025).
\newblock Data-driven output prediction and control of stochastic systems: An
  innovation-based approach.
\newblock \emph{Automatica}, 171, 111897.

\bibitem[{Wiener(1938)}]{Wiener1938}
Wiener, N. (1938).
\newblock The homogeneous chaos.
\newblock \emph{Am. J. Math.}, 60(4), 897--936.

\bibitem[{Willems et~al.(2005)Willems, Rapisarda, Markovsky, and
  De~Moor}]{Willems2005}
Willems, J.C., Rapisarda, P., Markovsky, I., and De~Moor, B.L.M. (2005).
\newblock A note on persistency of excitation.
\newblock \emph{Syst Control Lett}, 54(4), 325--329.

\bibitem[{Witteveen and Bijl(2006)}]{Witteveen2006}
Witteveen, J.A.S. and Bijl, H. (2006).
\newblock Modeling arbitrary uncertainties using {Gram-Schmidt} polynomial
  chaos.
\newblock In N.J. Pfeiffer (ed.), \emph{44th AIAA Aerospace Sciences Meeting
  and Exhibit}, 896. American Institute of Aeronautics and Astronautics Inc.
  (AIAA), United States.

\bibitem[{Yin et~al.(2024)Yin, Iannelli, and Smith}]{Yin2024}
Yin, M., Iannelli, A., and Smith, R.S. (2024).
\newblock Stochastic data-driven predictive control: Regularization,
  estimation, and constraint tightening.
\newblock \emph{IFAC-Pap.}, 58(15), 79--84.

\bibitem[{Özmeteler et~al.(2024)Özmeteler, Bilgic, Pan, Koch, and
  Faulwasser}]{Oezmeteler2024}
Özmeteler, M.B., Bilgic, D., Pan, G., Koch, A., and Faulwasser, T. (2024).
\newblock Data-driven uncertainty propagation for stochastic predictive control
  of multi-energy systems.
\newblock \emph{Eur. J. Control}, 80, 101066.

\end{thebibliography}
\endgroup
\appendix
                                      
\end{document}